\documentclass{amsart}

\usepackage{amsmath}
\usepackage{amssymb}
\usepackage{amscd}
\usepackage{cite}
\usepackage{todonotes}
\usepackage{graphicx}
\usepackage{url}
\usepackage{enumerate}

\newtheorem{theorem}{Theorem}
\newtheorem{lemma}[theorem]{Lemma}
\newtheorem{corollary}[theorem]{Corollary}
\newtheorem{proposition}[theorem]{Proposition}
\newtheorem{example}[theorem]{Example}
\newtheorem{definition}[theorem]{Definition}

\newcommand{\ZZ}{\mathbb Z}
\newcommand{\PP}{\mathbb P}
\newcommand{\fS}{\mathfrak S}

\newcommand{\autL}{A}
\newcommand{\iso}{\operatorname{iso}}
\newcommand{\newick}{\textsf}

\newcommand{\arxiv}[1]{#1}
\newcommand{\notarxiv}[1]{}

\newcommand{\FIGspr}{\
\begin{figure}
  \arxiv{\includegraphics[width=4in]{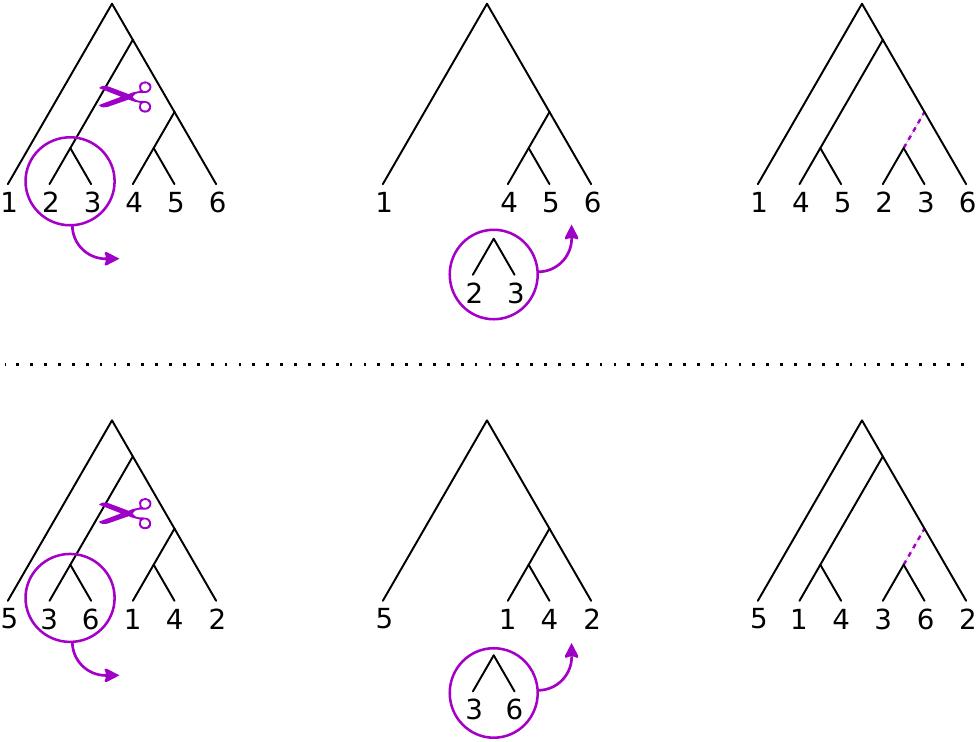}}
\caption{\
  Two equivalent subtree-prune-regraft moves applied to trees which are identical up to relabeling.
  The number of such moves required to transform one tree into another only depends on the \emph{relative} leaf labeling between the two trees.
}
\label{FIGspr}
\end{figure}
}

\newcommand{\FIGsprTanglegram}{\
\begin{figure}
  \arxiv{\includegraphics[width=2.5in]{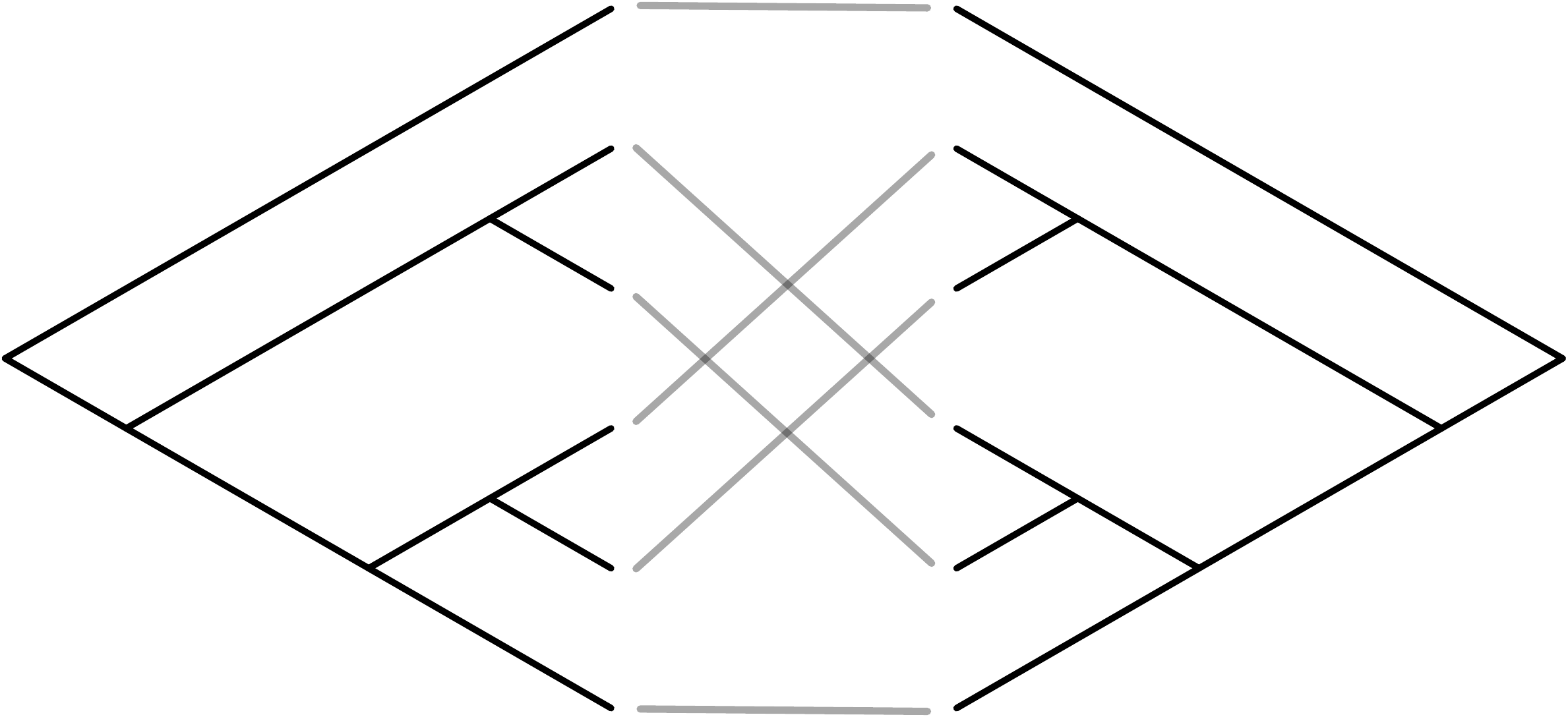}}
\caption{\
    The tanglegram corresponding to the pairs of trees in Figure~\ref{FIGspr}, with the bijection shown in gray.
    When considered as a graph, the black edges are called \emph{tree edges}, and the gray edges are called \emph{between-leaf edges}.
}
\label{FIGsprTanglegram}
\end{figure}
}

\newcommand{\FIGtanglegramFour}{\
\begin{figure}
  \arxiv{\includegraphics[width=2.5in]{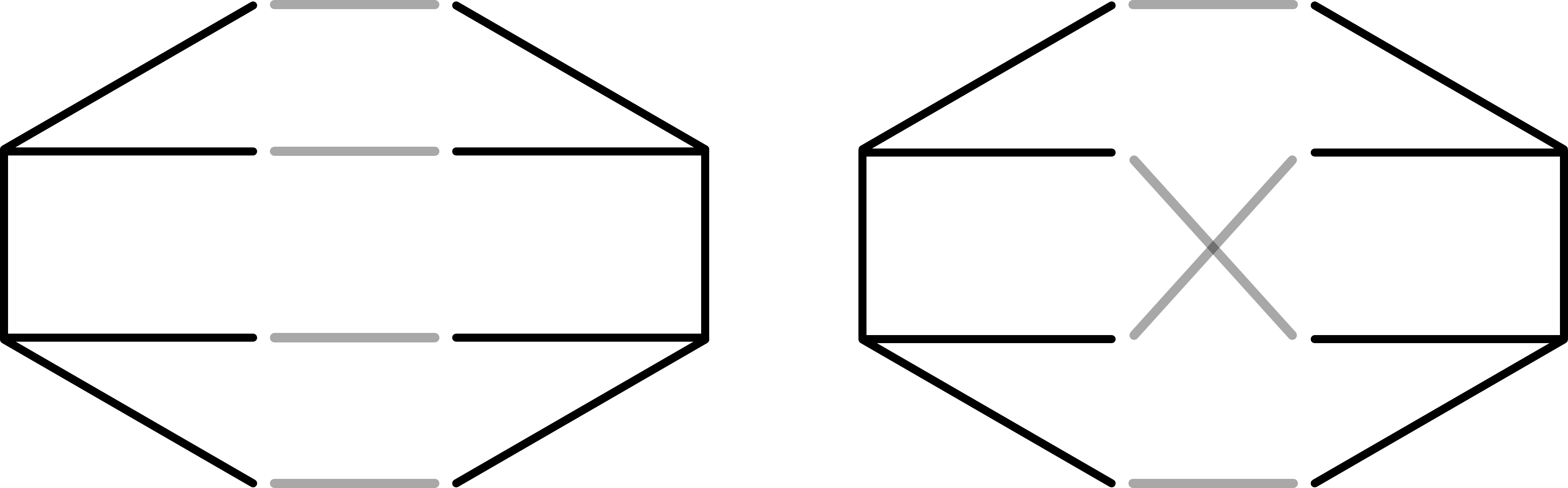}}
\caption{\
The two unrooted binary tanglegrams with four leaves.
}
\label{FIGtanglegramFour}
\end{figure}
}

\newcommand{\FIGnoFlip}{\
\begin{figure}
  \arxiv{\includegraphics[width=2in]{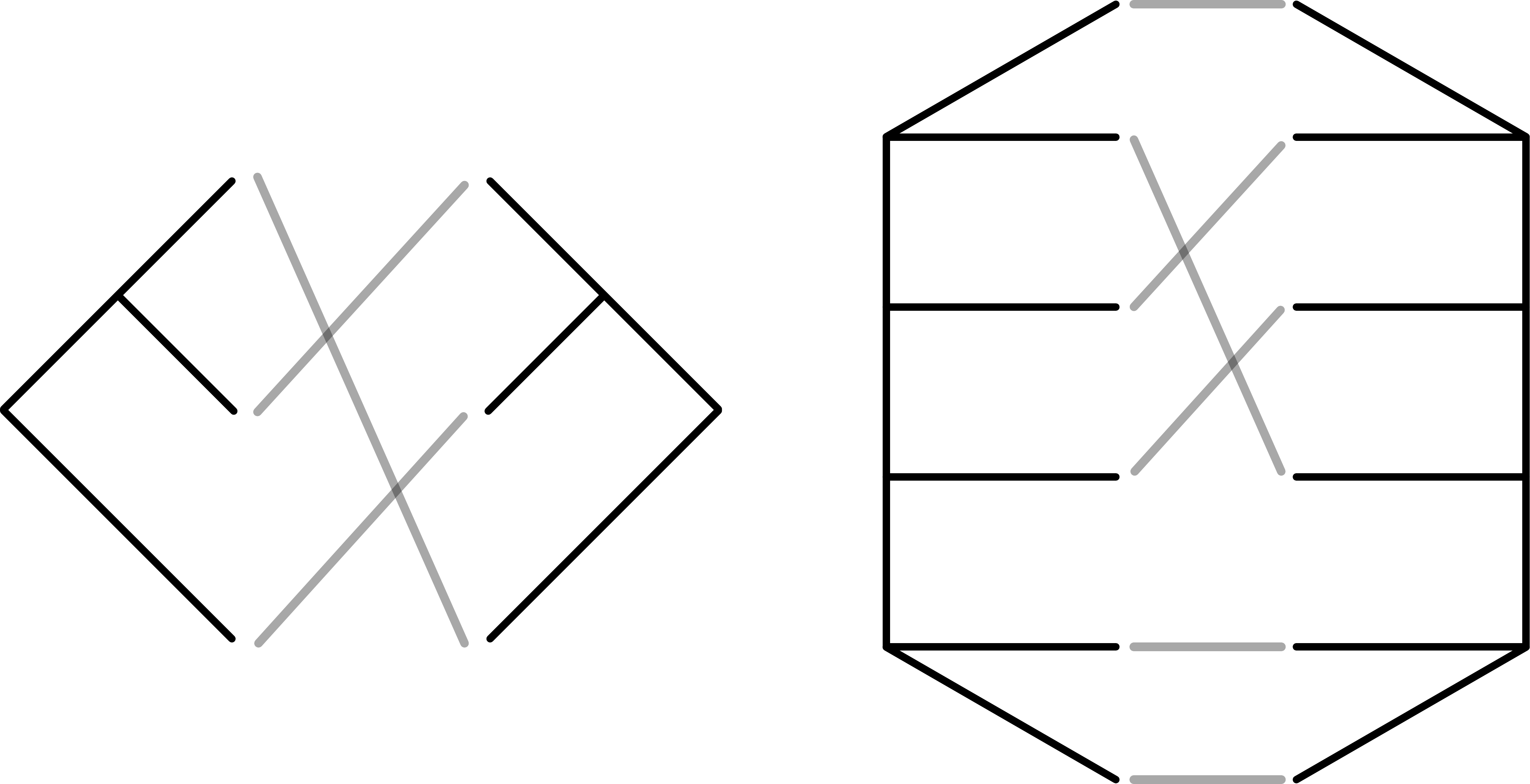}}
\caption{\
An ordered rooted and an unordered unrooted tanglegram formed by two copies of the same unrooted tree with no automorphism that switches the trees forming each tanglegram.
These examples show that the second condition of Proposition~\ref{prop:flip} is not always satisfied.
}

\label{FIGnoFlip}
\end{figure}
}

\newcommand{\FIGcounts}{\
\begin{figure}
  \arxiv{\includegraphics[width=.8\linewidth]{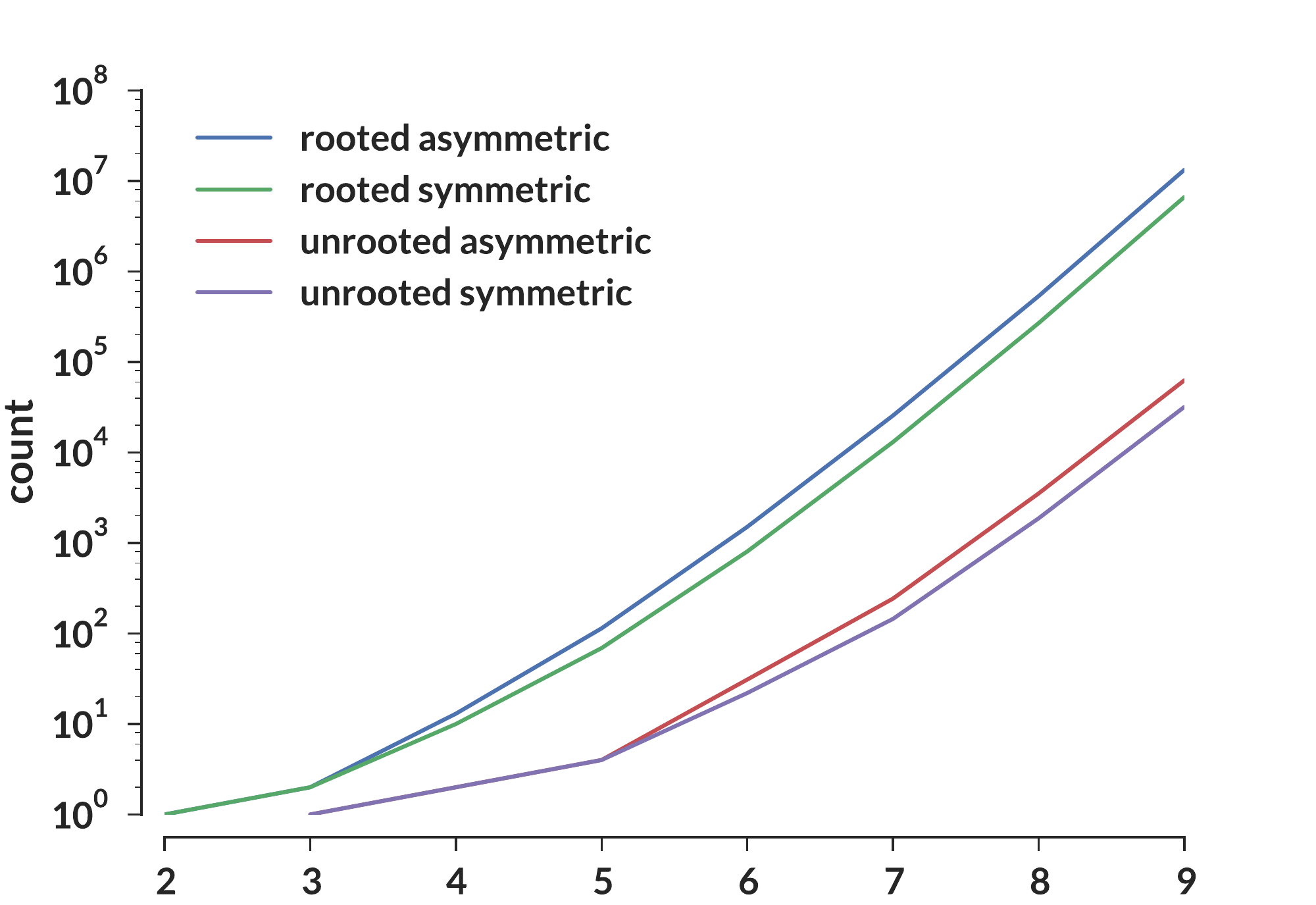}}
  \caption{Counts of various types of tanglegrams.}
\label{FIGcounts}
\end{figure}
}

\newcommand{\TABcounts}{\
\begin{table}
\centering
\small
\caption{
Enumeration of various types of binary tanglegrams.
These counts have been validated ``from below'' by checking for graph isomorphisms between exemplars for $n \leq 6$ in the rooted case, and $n \leq 7$ in the unrooted case.}
\arxiv{
\begin{tabular}{l|llll}
leaves & rooted ord. & rooted unord. & unrooted ord. & unrooted unord. \\
\hline
1      & 1             & 1            & 1               & 1              \\
2      & 1             & 1            & 1               & 1              \\
3      & 2             & 2            & 1               & 1              \\
4      & 13            & 10           & 2               & 2              \\
5      & 114           & 69           & 4               & 4              \\
6      & 1509          & 807          & 31              & 22             \\
7      & 25595         & 13048        & 243             & 145            \\
8      & 535753        & 269221       & 3532            & 1875           \\
9      & 13305590      & 6660455      & 62810           & 31929
\end{tabular}
}
\label{TABcounts}
\end{table}
}

\begin{document}
\title{Tanglegrams: a reduction tool for mathematical phylogenetics}

\author{Frederick A. Matsen IV}%
\address{Computational Biology Program,
Fred Hutchinson Cancer Research Center,
Seattle, WA 98109, USA}
\email{}%
\urladdr{http://matsen.fredhutch.org/}
\thanks{FAM partially supported by National Science Foundation grant DMS-1223057, SCB partially supported by National Science Foundation grant DMS-1101017, and MK supported by Research Program Z1-5434 and Research Project BI-US/14-15-026 of the Slovenian Research Agency.}

\author{Sara C.  Billey}%
\address{Department of Mathematics, University of Washington,
Seattle, WA 98195, USA}
\email{}%
\urladdr{http://www.math.washington.edu/~billey/}

\author[Kas]{Arnold Kas}%
\address{Computational Biology Program,
Fred Hutchinson Cancer Research Center,
Seattle, WA 98109, USA}

\author{Matja\v z Konvalinka}%
\address{Department of Mathematics,
University of Ljubljana,
Jadranska 21, Ljubljana, Slovenia}
\email{}%
\urladdr{http://www.fmf.uni-lj.si/~konvalinka/}
\thanks{}

\date{\today}

\begin{abstract}
Many discrete mathematics problems in phylogenetics are defined in terms of the relative labeling of pairs of leaf-labeled trees.
These relative labelings are naturally formalized as tanglegrams, which have previously been an object of study in coevolutionary analysis.
Although there has been considerable work on planar drawings of tanglegrams, they have not been fully explored as combinatorial objects until recently.
In this paper, we describe how many discrete mathematical questions on trees ``factor'' through a problem on tanglegrams, and how understanding that factoring can simplify analysis.
Depending on the problem, it may be useful to consider a unordered version of tanglegrams, and/or their unrooted counterparts.
For all of these definitions, we show how the isomorphism types of tanglegrams can be understood in terms of double cosets of the symmetric group, and we investigate their automorphisms.
Understanding tanglegrams better will isolate the distinct problems on leaf-labeled pairs of trees and reveal natural symmetries of spaces associated with such problems.
\end{abstract}

\maketitle

\section{Introduction}
Consider the problem of computing the \emph{subtree-prune-regraft} (SPR) distance between two leaf-labeled phylogenetic trees.
An SPR move cuts one edge of the tree and then reattaches the resulting rooted subtree at another edge (Figure~\ref{FIGspr}).
The SPR distance between two (phylogenetic, meaning leaf-labeled) trees $T_1$ and $T_2$ is the minimum number of SPR moves required to transform $T_1$ into $T_2$.
This distance is of fundamental importance in phylogenetics, and many papers have been written both applying \cite{Beiko2005-la,Whidden2014-yt} and investigating properties of \cite{Allen2001-yz,Bordewich2005-cx,Whidden2013-tf} this distance.

Say that we wanted to calculate the SPR distance between every pair of trees on a certain number of leaves.
Na\"ively this would require a large number of SPR calculations, namely the number of leaf-labeled phylogenetic trees choose two.
However, the distance between two such trees does not depend on the actual labels of $T_1$ and $T_2$, so one can permute the leaf labels without changing the distance.
Furthermore, a path made by intermediate trees between the two trees could also have its labels permuted in order to give a path between the trees with permuted leaf labels.
Thus, problems like SPR distance do not concern the actual leaf labels as such, but rather use the leaf labels as markers that can be used to map leaves of one phylogenetic tree on to another: the problem and its solutions are actually defined in terms of a \emph{relative} leaf labeling (Figure~\ref{FIGspr}).
\FIGspr

Analogous discrete mathematics problems and objects defined in terms of tuples of labeled combinatorial objects, but without direct reference to the labels themselves, are ubiquitous in computational biology.
Any distance between pairs of trees that is computed in terms of tree modifications, such as (rooted or unrooted) subtree-prune-regraft described above, \emph{nearest-neighbor-interchange} and \emph{tree bisection and reattachment} (see \cite{Bordewich2005-cx} for a review), satisfy this condition.
Such moves are used as the basis of both maximum-likelihood heuristic search and Bayesian Markov chain Monte Carlo (MCMC) tree reconstruction.
The corresponding graph, in which trees form vertices and a collection of moves form edges, has natural symmetries of pairs of points in these spaces which have the same relative labeling.
For example, hitting times of simple random walks on graphs formed by such moves for given start and end trees \cite{Aldous2000-vg,Diaconis2002-gy,Evans2006-xh} are defined in terms of relative labelings between the start and end trees.
The same is true for more complex random walks such as Markov chain Monte Carlo using a label-invariant likelihood, as would be used for sampling from a prior distribution on trees \cite{Alfaro2006-ru}.
Graph characteristics such as Ricci-Ollivier curvature \cite{Ollivier2009-bw} under simple random walks or MCMC with a label-invariant likelihood are expressed in terms of relative tree labelings \cite{Whidden2015-lg}.
Analogous considerations hold for the problem of species delimitation, which can naturally be phrased in terms of inference of a partition of relatively labeled objects: neither distances between partitions \cite{Gusfield2002-il} nor the graphs underlying MCMC over these partitions \cite{Yang2010-kc} actually refer to labels themselves.

The concept of a pair of rooted phylogenetic trees with a relative leaf labeling has been formalized as a \emph{tanglegram} \cite{Page1993-hr,Page2003-rr}.
A tanglegram is a pair of trees on the same set of leaves with a bijection between the leaves in the two trees \cite{Venkatachalam2010-zh} (Figure~\ref{FIGsprTanglegram}).
There has been extensive work on the problem of finding the layout of a given tanglegram in the plane that minimizes crossings, with the goal of most clearly visualizing co-evolutionary relationships between species \cite{Buchin2008-lc,Lozano2008-tp,Bansal2009-ni,Bocker2009-xl,Fernau2010-an,Venkatachalam2010-zh}.
\FIGsprTanglegram

However, we are not aware of any work considering tanglegrams as a convenient formalization of the notion of a relative leaf labeling in the context of studying pairs of labeled phylogenetic trees.
There has also been little work enumerating or finding other properties of tanglegrams until recently \cite{enumerateBinaryAsymmetric}.
In addition, more challenging and important problems in mathematical phylogenetics reduce to questions on relatively-labeled collections of more than two trees, and correspondingly one can extend the notion of tanglegram to more than two trees.
For example, ``supertree'' methods reconstruct a tree from collections of trees, each of which is typically considered to express information about the larger tree \cite{Bininda-Emonds2002-bt,Steel2008-pn,Whidden2014-yt}, which in fact is a problem on multi-tree tanglegrams.
The same is true for the minimal hybridization network \cite{Baroni2005-kr} and maximum agreement subtree \cite{Finden1985-bv,Farach1995-wm} problems.
Thus many problems in the discrete mathematics of phylogenetic trees ``factor'' through a problem concerning a generalized version of a tanglegram.

With this motivation for studying tanglegrams in more depth, here we formalize more general notions of tanglegram, describe their symmetries, observe that tanglegrams have a convenient algebraic formulation as double cosets of the symmetric group, and provide some enumeration results for four types of tanglegram.

\section{Tanglegrams}
An \emph{unrooted binary tree} $T$ is a finite graph for which there is a unique path between every pair of vertices, and such that every non-leaf vertex has degree three.
A \emph{rooted tree} is an unrooted tree with a distinguished node called the \emph{root}.
We will also make the assumption common in phylogenetics that the root of a rooted tree has degree two, and that there are no degree-two nodes other than the root (if there is a root).
The \emph{leaves} $L(T)$ of a tree $T$ are degree-one vertices of the tree.

\begin{definition}
Let $T$ and $S$ be trees with the same number of leaves.
An \emph{ordered tanglegram} $Y$ on $(T, S)$ is an ordered triple $(T, \phi, S)$, where $\phi$ is a bijection $L(T) \rightarrow L(S)$.
\end{definition}

The \emph{graph of the tanglegram} $Y$ is the graph formed from the union of $T$ and $S$ by adding an edge from each leaf $x$ in $T$ to the corresponding leaf $\phi(x)$ in $S$.
We will distinguish these \emph{between-leaf edges} from the \emph{tree edges} of $T$ and $S$ (Figure~\ref{FIGsprTanglegram}).

We have defined tanglegrams in terms of ordered triples $Y = (T, \phi, S)$ , so $Y' = (S, \phi^{-1}, T)$ is a different tanglegram.
This is a sensible definition when considering sequences of trees with an inherent directionality.
However, often there is not such a directionality, such as for subtree-prune-regraft moves, which are easily reversed.
This motivates the following concept:
\begin{definition}
A \emph{unordered tanglegram} is a pair $(\{T, S\}, \phi)$ where $\{T, S\}$ is an unordered set of two trees, and $\phi$ is a bijection between $L(T)$ and $L(S)$.
\end{definition}

\subsection{Automorphisms and tanglegram equivalence}
Let $V(X)$ denote the vertex set of a graph $X$.
An \emph{isomorphism} between unrooted trees $T$ and $S$ is a bijective map $h:V(T) \rightarrow V(S)$ in which $f$ maps edges of $T$ to edges of $S$.
For a rooted tree, we add the requirement that an isomorphism must map the root node of $T$ to the root node of $S$.
An \emph{automorphism} of a tree $T$ is an isomorphism of $T$ with itself.
It is clear that the degree of a node (i.e.\ the number of adjacent nodes) is preserved under isomorphisms.
In phylogenetics, it is common that the root of a tree is the only node of degree two.
In this case, there is no distinction between isomorphisms of rooted trees and isomorphisms of these trees as unrooted trees because degrees are preserved under isomorphism.

We start with an ``obvious'' lemma, the proof of which can be found in the Appendix.
First note that any isomorphism between trees $T$ and $S$ preserves the leaf sets $L(T)$ and $L(S)$, and therefore induces a bijection between $L(T)$ and $L(S)$.
\begin{lemma}
\pushQED{\qed}
An isomorphism between (rooted or unrooted) trees $T$ and $S$ is uniquely determined by the induced bijection between $L(T)$ and $L(S)$.
In particular, an automorphism of a tree $T$ is uniquely determined by the induced permutation of the leaf set $L(T)$.
\label{lemma:leafAction}
\popQED
\end{lemma}

Thus we will often consider an isomorphism as such a bijection $L(T) \rightarrow L(S)$.

\begin{definition}
\label{def:tangleEquiv}
Given two tanglegrams $Y = (T, \phi, S)$ and $Y' = (T, \phi', S)$ on the same pair of trees, an isomorphism of $Y$ and $Y'$ is defined by a pair of automorphisms $g: L(T) \rightarrow L(T)$, and $h: L(S) \rightarrow L(S)$ satisfying $h \circ \phi = \phi' \circ g$.
\end{definition}

The condition in the definition can be visualized in the commutative diagram
\[
\begin{CD}
L(T)    @>{\phi}>>   L(S)\\
@VV{g}V         @VV{h}V \\
L(T)   @>{\phi'}>>   L(S).
\end{CD}
\]

Note that if two tanglegrams $Y_1$ and $Y_2$ are isomorphic, then there is a 1-1 map from the graph of $Y_1$ to the graph of $Y_2$ which maps between-leaf edges to between-leaf edges.

\subsection{Symmetries of trees}
In order to describe the ensemble of tanglegrams it is necessary to review the symmetries of the trees in the tanglegram.
Although this material is classical, we were not able to find a simple presentation, and so provide one here.
We will assume familiarity with the basics of group theory (covered by dozens of textbooks, e.g.\ \cite{Dummit2004-ls}).
Automorphisms of a tree $T$ form a group under composition.
Using $\fS_n$ to denote the symmetric group on $n$ objects, leaf automorphisms of $T$ form a subgroup $\autL(T)$ of $\fS_{|L(T)|}$.

To enumerate symmetries of trees it is convenient to use the notion of a \emph{wreath product};
we will only define and use wreath product in the case when the acting group is $\fS_k$.
Use $G^k$ to denote the $k$-fold direct product $G \times \cdots \times G$.

Given a group $G$, the wreath product $G \wr \fS_k$ of $G$ by $\fS_k$ can be described as the direct product $G^k \times \fS_k$ with the following group operation.
First recall that the group operation on $G^k$ is defined by applying $G$'s group operation component-wise.
An element of $\fS_k$ acts on $G^k$ by permuting the components, such that the group action of $\sigma \in \fS_k$ on $g \in G^k$ is the element $\sigma(g) \in G^k$ with $i$th component $g_{\sigma(i)}$.
Given elements $g, g'$ in $G^k$ and $\sigma, \sigma' \in \fS_k$, the wreath group law is:
\[
(g, \sigma) \, (g', \sigma') := (g \, \sigma(g'), \sigma \sigma').
\]

For rooted trees, Jordan \cite{Jordan1869-ce} and P\'olya \cite{Polya1937-np} observed that the automorphism group of any rooted tree can be built by repeated direct products and wreath products of symmetric groups as follows.
In the simplest case, assume a rooted tree $T$ for which the root has two daughter subtrees $T_1$ and $T_2$.
If $T_1$ and $T_2$ are isomorphic (and thus have the same automorphism groups), the automorphism group of $T$ is the wreath product $\autL(T_1) \wr \fS_2$.
That is, its symmetry group is two copies of $\autL(T_1)$ along with the symmetry exchanging $T_1$ and $T_2$, equipped with the group operation that appropriately exchanges the subtrees before applying symmetries to the subtrees.
If $T_1$ and $T_2$ are not isomorphic, then $\autL(T)$ is simply the direct product $\autL(T_1) \times \autL(T_2)$.

Now let $T$ be a tree whose root has some number of daughters, each of which are roots of subtrees $T_1,\ldots,T_r$.
We can reorder and partition the subtrees into $N$ partitions:
\[
T_1,\ldots,T_{i_1}, \, T_{i_1 + 1},\ldots,T_{i_2}, \, \ldots, \, T_{i_{N-1} + 1}, \ldots,T_{i_N}
\]
such that the subtrees in each partition are isomorphic to one another and the subtrees in different partitions are not isomorphic.
This defines integers $i_1, \ldots, i_N$; take $i_0$ to be zero.
A more general version of the argument above establishes
\begin{theorem}[Jordan, 1869]
\pushQED{\qed}
$\autL(T)$ is the direct product $A_1 \times \cdots \times A_N$, where $A_j$ is the wreath product of $\autL(T_{i_j})$ with the symmetric group $\fS_{i_j - i_{j-1}}$.
\popQED
\label{theorem:treeAuto}
\end{theorem}
This defines the automorphism group of a rooted tree recursively, where of course the automorphism group of a single leaf is trivial.

\begin{example}
Let $T_n$ denote the perfectly balanced binary tree on $2^n$ leaves and let $G_n = \autL(T_n)$.
$G_2 = \fS_2$ and for each n, $G_n = G_{n - 1} \wr \fS_2$.
Moreover, $|G_n| = 2 |G_{n - 1}|^2$.
\end{example}

\begin{example}
The symmetry group of the Newick-format \cite{wiki:newick} tree \newick{(1,((2,3),((4,5),6)));} (shown as the upper-left tree of Figure~\ref{FIGspr}) is the direct product of the symmetry groups of $(2,3)$ and $((4,5),6)$.
Each of these symmetry groups are $\fS_2$.
\end{example}

The automorphism group of an unrooted tree will become clear after we describe a classical and mathematically natural way to root an unrooted tree: at the \emph{centroid}.
Let $T$ be a tree, and let $x$ be a node of $T$.
If we remove $x$ as well as the edges attached to $x$ from $T$, we obtain a number of disjoint connected and rooted subtrees, $X_1,\ldots, X_k$.
\begin{definition}
The weight of $x$, $w(x)$, is defined as the maximum number of nodes of the subtrees $X_1, \ldots, X_k$.
\end{definition}
\begin{definition}
The node $x$ is said to be a centroid of $T$ if $w(x)$ is minimal over all nodes of $T$.
\end{definition}

It is clear that any automorphism of $T$ maps a centroid to a centroid, a fact which we will use to find a root fixed under leaf automorphism.
Centroids are unique or nearly so, as shown by the following theorem, the proof of which can be found as a guided exercise in \cite[\S 2.3.4.4]{Knuth1973}.
\begin{theorem}[Jordan, 1869]
\label{thm:jordan}
\pushQED{\qed}
Every tree has either:
\begin{enumerate}[1.]
\item a unique centroid or
\item two adjacent centroids.
\end{enumerate}
In case 2, every automorphism either preserves the centroids or exchanges them.
\popQED
\end{theorem}

Let $T$ be an unrooted tree, and let $T_r$ be the rooted tree formed by rooting $T$ at either the unique centroid, or by a new node in the edge joining a pair of centroids.
\begin{corollary}
\pushQED{\qed}
The automorphism group of an unrooted tree $T$ is identical to the automorphism group of the associated rooted tree $T_r$.
\popQED
\end{corollary}

\begin{example}
The symmetry group of the six-leaf unrooted tree with three two-leaf subtrees (Newick format \newick{((1,2),(3,4),(5,6));}) is $\fS_2 \wr \fS_3$.
\end{example}

\subsection{Double cosets and enumeration of tanglegrams}
We are now ready to algebraically describe the set of tanglegrams on a pair of $n$-leaf trees.
Assume $n$-leaf trees $T$ and $S$, which are both rooted or both unrooted.
Arbitrarily mark the elements of the leaf sets $L(T)$ and $L(S)$ with the same set of $n$ symbols, such that we can identify both $\autL(T)$ and $\autL(S)$ as subgroups of $\fS_n$.
Using this same marking, we can also think of the bijections from $L(T)$ to $L(S)$ as being elements of $\fS_n$, thus these elements of $\fS_n$ give tanglegrams on $T$ and $S$.
Recall Definition~\ref{def:tangleEquiv}, stating that the set of bijections $\phi'$ giving the same tanglegram as a given $\phi$ are those for which there exist automorphisms $g \in \autL(T)$ and $h \in \autL(S)$ such that $h \circ \phi = \phi' \circ g$.
This criterion is equivalent to $\phi' = h \phi g^{-1}$ as group elements in $\fS_n$.
The set of elements satisfying such a criterion is called a double coset \cite{Dummit2004-ls}.

\begin{definition}
Given a subgroup $J$ of a group $G$ and $g \in G$, the \emph{right coset} $Jg$ (resp. \emph{left coset} $gJ$) $G$ is the set of elements of the form $\{jg \mid j \in J\}$ (resp. $\{gj \mid j \in J\}$).
The number of right cosets of $J$ in $G$ is equal to the number of left cosets. This number is defined as the index of $J$ in $G$ and is denoted $[G:J]$.
Given two subgroups $J$ and $K$ of $G$, the \emph{double coset} $JgK$ for some $g \in G$ is the set of elements $\{jgk \mid j \in J, k \in K\}$.
\end{definition}
Any two right (left) cosets of $J$ in $G$ are either identical or disjoint and the number of elements in any coset is the same, i.e. $|J|$.
In contrast to single cosets (left or right), the number of elements in a double coset may vary.
We state these observations, and the equivalent observations in the unordered case, as a proposition.
\begin{proposition}
\pushQED{\qed}
\label{prop:cosetTanglegramEquivalence}
Given two trees $T$ and $S$ with $n$ leaves,
\begin{itemize}
\item
the set of tanglegrams isomorphic to a tanglegram $(T, w, S)$ is in 1-1 correspondence with the double coset $\autL(S) w \autL(T)$ of $\fS_n$.
\item
the set of unordered tanglegrams isomorphic to $(\{T, S\}, w)$ is in 1-1 correspondence with equivalence classes of double cosets $\autL(S) w \autL(T)$ where pairs of cosets $HwK$ and $Kw^{-1}H$ are deemed equivalent.
\end{itemize}
\popQED
\end{proposition}
Note that the actual 1-1 correspondence depends on the marking of $T$ and $S$.

Here are some useful facts concerning cosets \cite{Dummit2004-ls,Lang-Algebra}:
\begin{itemize}
\item any two cosets are either disjoint or identical
\item every double coset is a disjoint union of right cosets and a disjoint union of left cosets
\item the number of right cosets of $H$ in $HgK$ is the index $[K:K \cap g^{-1}Hg]$,
and the number of left cosets of $K$ in $HgK$ is the index $[H:H \cap gKg^{-1}]$.
\end{itemize}

Combining these facts with the proposition above, we get:
\begin{proposition}
\pushQED{\qed}
The number of bijections from $L(T)$ to $L(S)$ giving an ordered tanglegram isomorphic to $Y = (T, w, S)$ is equal to
$| \autL(S) | [\autL(T) : \autL(T) \cap w^{-1}\autL(S)w]$, or equivalently
$| \autL(T) | [\autL(S) : \autL(S) \cap w \autL(T) w^{-1}]$.
\popQED
\end{proposition}

\FIGtanglegramFour

\begin{example}
Let $T$ and $S$ be the unique binary unrooted tree with 4 leaves.
There are two distinct tanglegrams on $(T, S)$ in both the ordered and unordered cases (Figure~\ref{FIGtanglegramFour}).
The automorphism group of either tree, $\autL(T)$, is the wreath product of $\fS_2$ by $\fS_2$, thus of order 8 (set theoretically $\ZZ_2 \times \ZZ_2 \times \ZZ_2$).
Marking the leaves with the integers 1 through 4 such that $(1,2)$ and $(3,4)$ are both sister pairs, $G = \autL(T)$ is generated by $\{ (12), (34), (13)(24) \} \subset \fS_4$.

The symmetric group $\fS_4$ contains $4! = 24$ elements.
Every double coset is a disjoint union of single cosets, and $G$ contains 8 elements, therefore the number of elements in a double coset is a multiple of 8.
Moreover, since the double cosets partition $\fS_4$, we either have 3 double cosets (each of 8 elements), or 2 double cosets (one of 8 elements and one of 16 elements), or one coset (of 24 elements).
Taking $w = (23)$, we calculate:
\[
G \cap w G w^{-1} = \{(), (12)(34), (13)(24), (14)(23)\}.
\]
Using the properties of double cosets, we find that the number of single cosets in the double coset $G w G$ is the index
$[G:G \cap w^{-1} G w] = 2$.
Thus this double coset has 16 elements, and so there must be two double cosets, corresponding to the two tanglegrams.
\end{example}

\subsection{Symmetries of tanglegrams}
\begin{definition}
An automorphism of an ordered tanglegram $Y$ is an automorphism of the graph of $Y$ which maps each tree to itself.
An automorphism of an unordered tanglegram $Y$ is an automorphism of the graph of $Y$ which preserves the between-leaf edges, so an automorphism of an unordered tanglegram either maps each tree to itself or switches the two trees.
If $Y$ is a rooted tanglegram, then an automorphism of $Y$ is required to preserve the roots of the two trees.
\end{definition}
If the automorphism $f:Y \rightarrow Y$ exchanges the two trees, $f$ is described by a pair of isomorphisms: $g_1:T \rightarrow S$ and $g_2:S \rightarrow T$.
For any leaf $x$ of $T$, the image of a bijective pair $(x, \phi(x))$ must map to another bijective pair $(g_2(\phi(x)), g_1(x))$.
This implies that $g_1(x) = \phi(g_2(\phi(x)))$, and thus in general that $g_1 = \phi \circ g_2 \circ \phi$.
If we put the same set of distinguishing marks on the leaves of the trees $T$ and $S$, we may consider the bijection $\phi$ to be an element of the symmetric group $\fS_n$.
With these conventions, we have shown that there exist $g_1 \in \autL(T)$ and $g_2 \in \autL(T)$ such that $g_1 = \phi \, g_2 \, \phi$ as group elements when there is an automorphism that switches the two trees.
The converse follows from reversing this argument.
In summary:
\begin{proposition}
\label{prop:flip}
\pushQED{\qed}
If $Y$ is an unordered tanglegram, then there exists an automorphism of $Y$ that switches the two trees if and only if:
\begin{itemize}
\item
the trees $T$ and $S$ are isomorphic, and
\item
$\phi \autL(T) \phi \cap \autL(T) \neq \emptyset$.
\end{itemize}
\popQED
\end{proposition}

On the other hand, if $h:Y \rightarrow Y$ is an automorphism which maps each tree to itself, then $f$ is described by two automorphisms $g:T \rightarrow T$ and $h:S \rightarrow S$ satisfying $\phi \circ g = h \circ \phi$ when restricted to the leaves, or $g = \phi^{-1} h \phi$ as elements of the symmetric group.
\begin{proposition}
\pushQED{\qed}
Assume an ordered tanglegram $Y = (T, \phi, S)$, or an unordered tanglegram $(\{T,S\}, \phi)$.
Set $H = \autL(T) \cap \phi^{-1} \autL(T) \phi$.
\begin{enumerate}[1.]
\item
If $Y$ is ordered or $T$ is not isomorphic to $S$, then $\autL(Y) = H$.
\item
If $Y$ is unordered and $T$ is isomorphic to $S$, then,
\begin{enumerate}[a.]
\item
if $\autL(T) \cap \phi \autL(T) \phi \neq \emptyset$, then
$\autL(Y)$ contains $H$ as a subgroup of index 2.
\item
otherwise, $\autL(Y) = H$.
\end{enumerate}
\end{enumerate}
\popQED
\end{proposition}

\FIGnoFlip

Similar to the case for trees, tanglegram automorphisms are determined entirely by their action on the leaves of one of the trees.

\subsection{Labeled tanglegrams}
Analogous to the concept of a leaf-labeled tree, there is a concept of a labeled tanglegram.
\begin{definition}
A \emph{labeled tanglegram} is a tanglegram along with a bijective map of the label set $X$ to the leaves of one of the trees.
\end{definition}
This is analogous to the definition of a leaf-labeled phylogenetic tree \cite{Semple2003-em}.
The other tree can be considered to be labeled by the composition of the labeling with the bijection.
Applying this labeling to both trees and then forgetting the bijection gives a pair of leaf labeled trees on the same label set, and each such pair of leaf labeled trees obviously determines a labeled tanglegram.
Thus, labeled tanglegrams are in one-to-one correspondence with pairs of leaf-labeled phylogenetic trees.
If the tanglegram is ordered, then this is an ordered pair of trees, and if unordered it is unordered.

It is natural to ask how many distinct labeled $n$-tanglegrams have the same underlying ordered or unordered tanglegram.
Each leaf has a distinct label, such that the symmetric group acts freely on these labels.
By the orbit-stabilizer theorem,
\begin{proposition}
\label{prop:nlabelings}
\pushQED{\qed}
The number of leaf-distinct labelings of a given $n$-tanglegram $Y$ is equal to $n! / |\autL(Y)|$.
\popQED
\end{proposition}
This is true for ordered and unordered tanglegrams, using their respective automorphism definitions.
For example, there are 12 labelings for the ordered tanglegram \newick{(1,(2,(3,4))); (((1,2),3),4);} but only 6 when considered as an unordered tanglegram.

Given a means of sampling uniformly from tanglegrams \cite{enumerateBinaryAsymmetric}, we can use this proposition to obtain a weighted sampling scheme for the uniform distribution across pairs of phylogenetic trees on the same labeling set.
For example, assume we wanted to approximate the expectation of a function $f$ on uniformly sampled pairs of labeled trees, but which is constant on pairs of trees that make the same tanglegram (such as SPR distance).
Then
\[
\sum_{T_1, T_2} f(T_1, T_2) \PP(T_1, T_2) = \sum_{Y} f(T_1, T_2) \PP(T_1, T_2 | Y) \PP(Y)
\]
where if $f(T_1, T_2) = f(T_2, T_1)$ for all $T_1, T_2$ then the right hand sum can be over unordered tanglegrams $Y$, and otherwise it is over ordered tanglegrams $Y$.
Here $\PP(T_1, T_2 | Y)$ is simply the indicator function expressing if $T_1$ and $T_2$ make $Y$, divided by the number of pairs of labeled trees making $Y$ as enumerated in Proposition~\ref{prop:nlabelings}.
Rather than sampling pairs of trees uniformly and calculating an empirical expectation as on the left side, we can get a lower variance estimator by sampling tanglegrams uniformly and weighting them as on the right hand side.
Such a means of sampling uniformly from tanglegrams in the rooted binary ordered case is given in \cite{enumerateBinaryAsymmetric}.

\section{Variants and special cases}

\subsection{Multiple trees}
The definition of a tanglegram on two trees can be generalized to a version on multiple trees.
\begin{definition}
Given trees $T_1, \ldots, T_n$ with the same number of leaves, a \emph{multi-tanglegram} on this set of trees is given by a pair of tuples $((T_1, \ldots, T_n), (\phi_{ij})_{i,j \in 1, \ldots, n})$ in which $\phi_{ij}:L(T_i) \rightarrow L(T_j)$ are bijections satisfying:
\begin{enumerate}[1.]
\item $\phi_{ii} = 1$ for all i;
\item $\phi_{ji} = \phi_{ij}^{-1}$ for all i, j;
\item $\phi_{ik} = \phi_{jk} \circ \phi_{ij}$, for all i, j, k.
\end{enumerate}
\end{definition}

We can also generalize the definition of isomorphism to multi-tanglegrams on $n$ trees.
\begin{definition}
Two multi-tanglegrams $Y = ((T_1, \ldots, T_n), (\phi_{ij})_{i,j \in 1, \ldots, n})$ and
$Y' = ((T_{1}, \ldots, T_{n}), (\phi'_{ij})_{i,j \in 1, \ldots, n})$ on the same list of trees
are isomorphic if there exist automorphisms
$(g_i:T_i \rightarrow T_i)_{i \in 1, \ldots, n}$ and
$(h_i:T_i \rightarrow T_i)_{i \in 1, \ldots, n}$ satisfying
$h_j \circ \phi_{ij} = \phi'_{ij} \circ g_i$ for $i,j = 1, \ldots, n$.
\end{definition}

It is clear that the $n^2$ bijections $\phi_{ij}$ are completely determined
by the $n - 1$ bijections $\{\phi_{1i}\}_{i = 2, \ldots, n}$, since
$\phi_{ij} = \phi_{1j} \circ \phi_{1i}^{-1}$.
With this observation, we can rephrase the definition of isomorphism above, which we will state as a proposition:
\begin{proposition}
Using the notation above, multi-tanglegrams $Y_1$ and $Y_2$ are isomorphic if and only if there exist automorphisms $g_i \in \autL(T_i), i = 1, \ldots, n$ satisfying
$\phi_{1i}' = g_i \circ \phi_{1i} \circ g_{1}^{-1}$.
\end{proposition}
Alternatively, the automorphisms $\phi_{ij}$ are completely determined by a sequence $\phi_{12}, \phi_{23}, \ldots, \phi_{k-1\, k}$, and thus multi-tanglegrams are called \emph{tangled chains} by \cite{enumerateBinaryAsymmetric}.

\subsection{More general classes of graphs}
Another direction of generalization involves considering more general classes of graphs.
For example, the tanglegram layout problem has been studied for rooted phylogenetic networks \cite{Scornavacca2011-xa}.
Given a natural number $n$, define an \emph{$n$-leaved graph} as a graph $U$ along with $n$ distinguished vertices $L(U)$ called \emph{leaves}.
\begin{definition}
Given a natural number $n$, define a \emph{generalized $n$-tanglegram} as a triple $(U, \phi, V)$, where $U$ and $V$ are a pair of $n$-leaved graphs and $\phi$ is a bijection between $L(U)$ and $L(V)$.
\end{definition}

Equivalent statements to those above can also hold in this more general setting.
If we require that $n$-leaved graph automorphisms preserve the leaf set $L(U)$, we can again define the leaf automorphism group $\autL(U)$ to be the automorphism group of $U$ restricted to $L(U)$.
If the graphs are such that any graph automorphism is determined by its action on the leaf set, then generalized tanglegrams on a given pair of $n$-leaved graphs $U$ and $V$ are in one-to-one correspondence with double cosets $\autL(V) w \autL(U)$ in $\fS_n$.

\subsection{Partitions}
Another line of inquiry in computational evolutionary biology concerns species delimitation, which can naturally be phrased in terms of inference of a partition of labeled objects.
In a manner analogous to phylogenetic trees, researchers use MCMC to explore the posterior on such partitions \cite{Yang2010-kc}, and comparison of the results can be performed using distances between the partitions \cite{Gusfield2002-il}.
Similar considerations hold for random walks and these distances as described in the introduction for trees.
These partitions can also be thought of as a certain type of leaf-labeled tree of height two, thus pairs of partitions on the same underlying set also give a type of tanglegram.

All of the above conclusions hold for such partition tanglegrams as well.
The automorphisms of a partition are a special case of Theorem~\ref{theorem:treeAuto}.
For example, the partition $123 \mid 456 \mid 78$ has automorphism group $(\fS_3 \wr \fS_2) \times \fS_2$.

\section{Enumeration}
Using a computer algebra package such as GAP4 \cite{GAP4} which is able to enumerate double cosets, and a package such as Sage \cite{SteinJoyner2005} which can obtain symmetry groups of graphs, one can apply Proposition~\ref{prop:cosetTanglegramEquivalence} to directly enumerate any type of tanglegram on a given pair of trees.
We have provided code to enumerate and work with tanglegrams at \url{https://github.com/matsengrp/tangle}.

For the case of binary ordered rooted tanglegrams, an elegant formula for the total number of tanglegrams on $n$ leaves $t_n$ has recently been found \cite{enumerateBinaryAsymmetric}.
One can use this formula, along with the number of tanglegrams on pairs of isomorphic trees, to compute the number of unordered tanglegrams as follows.

An unordered tanglegram is represented twice in the list of ordered tanglegrams on $n$ leaves if the two trees are non-isomorphic, or if the trees are isomorphic and the coset is different when the representative is inverted as in Figure~\ref{FIGnoFlip}.
For $n$ leaves, we let $s_n$ be the number of unordered tanglegrams, and then let $t_n^{\iso}$ be the number of ordered tanglegrams and $s_n^{\iso}$ the number of unordered tanglegrams on isomorphic pairs of trees.
To get $s_n$, we start with $t_n$ and subtract off half the number of ordered tanglegrams on non-isomorphic trees for the first case, and then subtract off $t_n^{\iso} - s_n^{\iso}$ for the second.
Simplifying $t_n - (t_n - t_n^{\iso})/2 - (t_n^{\iso} - s_n^{\iso})$, we get $s_n = \left(t_n -t_n^{\iso}\right)/2 + s_n^{\iso}$ for any $n \geq 3$.
\FIGcounts
\TABcounts

Such direct enumeration of various types of tanglegrams (Figure~\ref{FIGcounts}, Table~\ref{TABcounts}) suggests that their number grows super-exponentially.
In fact, that the number of (binary ordered rooted) tanglegrams is $O(n!\,4^n\,n^{-3})$ as shown by \cite{enumerateBinaryAsymmetric}.

There are thus many fewer such tanglegrams than there are pairs of leaf-labeled trees.
Indeed, a simplification of the argument establishing Corollary~8 of \cite{enumerateBinaryAsymmetric} shows that the ratio of the number of ordered pairs of leaf-labeled rooted trees to the number of binary ordered rooted tanglegrams is asymptotically a constant times the order of the symmetric group:
\[
\frac{\left((2n-3)!!\right)^2}{t_n} \sim \frac{n!}{e^{1/8}}.
\]
Intuitively, although the action of the symmetric group is not always free, ``for most cases it is close'' to free.
This may suggest that for $n$ leaves, the ratio of the number of ordered pairs of leaf-labeled unrooted trees to the number of binary ordered unrooted tanglegrams is also of order $n!$.

\section{Discussion}
Tanglegrams have been an object of study since before DNA sequences were widely available for the reconstruction of phylogenetic trees \cite{Hafner1988-da}.
So far they have been studied before in the context of co-evolutionary analyses, classically that between a host and a parasite, a subject of continuing interest \cite{Libeskind-Hadas2009-tx,Drinkwater2014-yr}.
As such, there has been extensive work on the case in which two rooted trees are distinguished between one another, as when one tree represents hosts and one parasites, which we call the ordered rooted case.
Here we have broadened the definition of tanglegrams by considering a broader class of underlying graphs, including unordered and/or unrooted tanglegrams.

In this form, tanglegrams formalize statements concerning pairs of phylogenetic trees on the same leaf set that do not directly make reference to the labels themselves.
Symmetric tanglegrams also do not make reference to the order of the trees.
We observe that many problems in phylogenetic combinatorics ``factor'' through a problem on tanglegrams.
As such, we believe tanglegrams to be a worthwhile object of study in phylogenetic combinatorics, and note that they have already been crucial in an analysis of the geometry of the subtree-prune-regraft graph \cite{Whidden2015-lg}.

These generalized notions of tanglegrams, which are equivalent to the collection of double cosets formed by the automorphism groups of the two trees, invite further investigation by combinatorialists.
An elegant formula for the number of binary ordered rooted tanglegrams has recently been found \cite{enumerateBinaryAsymmetric}, as well as for the multi-tanglegram case.
Here we provide the first several terms of the analogous sequence for unordered and/or unrooted tanglegrams; Ira Gessel has used the theory of species to develop means to enumerating unordered tanglegrams, which will be described in a forthcoming paper \cite{gessel.2015}.
It would be helpful to have a means of efficiently sampling other classes of tanglegrams according to familiar distributions on labeled phylogenetic trees, perhaps building on the method of sampling binary ordered rooted tanglegrams uniformly at random in \cite{enumerateBinaryAsymmetric}.

\section{Acknowledgements}
We would like to thank
Steve Evans,
Ira Gessel,
Michael Landis,
Chris Whidden,
and Bianca Viray.
We also thank the authors of the Sage and GAP4 software, especially Alexander Hulpke.

\bibliographystyle{ieeetr}
\bibliography{tangletex}

\clearpage

\section*{Appendix}

\begin{proof}[Proof of Lemma~\ref{lemma:leafAction}]
Assume that two isomorphisms $f, g:T \rightarrow S$ induce the same bijection of $L(T)$ to $L(S)$.
Let $\Pi$ be a path between any two leaves of $T$.
Then $f(\Pi)$ and $g(\Pi)$ are paths are paths between the same leaves of $S$ and thus are identical by definition of a tree.
Now, we just need to prove that every internal vertex $x$ lies on a path joining two leaves.
Since $x$ is internal, it belongs to at least two edges $(x, y)$ and $(x, y')$.
Consider a sequence of vertices obtained by following edges in the graph without backtracking, i.e.\ such that $(w,z)$ never follows $(z,w)$, starting with $(x,y)$.
Because the tree is finite and contains no loops by definition, this path will terminate at a leaf.
The same argument applied to $(x,y')$ finds another leaf such that the path between these two leaves contains $x$.
\end{proof}

\end{document}